\documentclass[11pt,a4paper]{amsart}

\usepackage{amsmath,amssymb,mathtools}
\usepackage[T1]{fontenc}
\usepackage[utf8]{inputenc}
\usepackage{enumitem}
\usepackage[dvipsnames]{xcolor}
\usepackage{cancel}
\usepackage{comment}

\numberwithin{equation}{section}

\newtheorem{theorem}{Theorem}[section]
\newtheorem{lemma}[theorem]{Lemma}
\newtheorem{definition}[theorem]{Definition}
\newtheorem{remark}[theorem]{Remark}
\newtheorem{corollary}[theorem]{Corollary}

\newtheorem{observation}[theorem]{Observation}
\newtheorem{conjecture}[theorem]{Conjecture}

\newtheorem{proposition}[theorem]{Proposition}

\DeclareMathOperator{\End}{End}
\DeclareMathOperator{\Id}{Id}

\DeclareMathOperator{\spanop}{span}

\newcommand{\K}{\mathbb K}

\def\reals{\mathbb{R}}
\def\compl{\mathbb{C}}

\def\Tr{\mathrm{Tr}}

\def\O{\mathcal{O}_M}

\def\T{\mathcal{T}_M}

\def\reals{\mathbb{R}}
\def\compl{\mathbb{C}}

\def\Tr{\mathrm{Tr}}

\def\O{\mathcal{O}_M}

\def\T{\mathcal{T}_M}

\newcommand{\eqa}{\begin{eqnarray}}
\newcommand{\eeqa}{\end{eqnarray}}
\newcommand{\beq}{\begin{equation}}
\newcommand{\eeq}{\end{equation}}

\newcommand{\w}{\omega}

\newcommand{\nav}{\nabla^{(A,v)}}
\newcommand{\Cyc}{\mathcal{Z}}

\begin{document}

\title[Haantjes torsion and integrability: a proof of BKM conjecture]{Haantjes torsion and integrability: a proof of Bolsinov-Konyaev-Matveev's conjecture}
\author{Alessandro Arsie}
\address{Department of Mathematics and Statistics, The University of Toledo, 2801 W. Bancroft St., Toledo, OH 43606, USA}
\email{alessandro.arsie@utoledo.edu}

\author{Paolo Lorenzoni}
\address{Dipartimento di Matematica e Applicazioni, Universit\`a degli Studi di Milano-Bicocca, Via Roberto Cozzi 53, I-20125 Milano, Italy; and INFN Sezione di Milano-Bicocca}
\email{paolo.lorenzoni@unimib.it}

\date{}
\maketitle

\begin{abstract}
We prove a conjecture formulated by Bolsinov,  Konyaev and Matveev in \cite{BKMint} stating that integrability of a system of hydrodynamic type ${\bf u}_t=A({\bf u}) {\bf u}_x$ with $\mathfrak{gl}$-regular $A$ at a point $p$ implies the vanishing of the Haantjes tensor of $A$ and of all its symmetries in a neighborhood of $p$.   As a consequence,  using the result of \cite{BKM3},  in a neighborhood of an algebraically generic point, any integrable system of hydrodynamic type defined by a $\mathfrak{gl}$-regular operator field can be written as ${\bf u}_t=X({\bf u})\circ {\bf u}_x$ where $X$ is a vector field and $\circ$ is a commutative associative product satisfying Hertling-Manin conditions.
\end{abstract}

\section{Introduction}\label{SectionIntro}
Let 
\begin{equation}\label{SHTintro}	
{\bf u}_t=A({\bf u}){\bf u}_x
\end{equation}
be a system of quasilinear first order evolutionary PDEs (system of hydrodynamic type). It is well-known that the entries $A^i_j({\bf u})$ of the matrix $A$ defining the system behave as the components of a $(1,1)$ tensor field when  changing the  dependent variables of the system. In the strictly hyperbolic case Tsarev proved that sufficient conditions for integrability are
\begin{enumerate}
\item The vanishing of the Haantjes tensor associated with $A$, i.e. the existence of Riemann invariants reducing system \eqref{SHTintro} to the form 
\begin{equation}	
{\bf r}_t=V({\bf r}){\bf r}_x
\end{equation}
with $V=\text{diag}(v^1,....,v^n)$.
\item The semi-Hamiltonian condition
\begin{equation}
	\label{tsarev1}
	\partial_j\Gamma^i_{ik}=\partial_k\Gamma^i_{ij},\qquad i\ne j\ne k\ne i.
\end{equation}
where
\begin{equation}\label{ChS}
	\Gamma^i_{ij}:=\frac{\partial_j v^i}{v^j-v^i},\qquad i\ne j.
\end{equation}
Sevennec proved in \cite{Sevennec} that semi-Hamiltonian systems coincide with conservative systems of hydrodynamic type admitting Riemann invariants. Conditions
 \eqref{tsarev1} are  automatically satisfied in the Hamiltonian case confirming a conjecture by Novikov about the integrability of diagonalizable Hamiltonian systems of hydrodynamic type.
\end{enumerate}
In a recent paper \cite{BKMint} the authors proved that the Haantjes condition is also necessary for integrability in the diagonal  case.
\newline
\newline
The study of non-diagonalizable systems is more recent:
\begin{itemize}
\item In \cite{LM} the authors considered systems defined by a $(1,1)$-tensor field of the form
\begin{equation}\label{LM}
A=L-a\,I
\end{equation}
where $L$ is a Nijenhuis operator, $I$ is  the identity operator and $a$ is a function satisfying the cohomological equation
\[d\cdot d_L\,a=0,\]
where $d_L$ is the Fr\"{o}licher-Nijenhuis differential. 
This construction does not require the existence of Riemann invariants. For instance, in the case of Nijenhuis operators coming from regular F-manifolds, these systems have been  studied  in \cite{LPVG3} in the general case (arbitrary  number of Jordan blocks and Jordan  blocks of arbitrary size) using a special set of coordinates introduced by  David  and Hertling  in  \cite{DH}.
\item In \cite{LPR,AL13,LPVG1,LPVG2, LPVG3,ALPVG,AL} the authors considered systems of the form (F-systems)
\begin{equation}\label{F-sys-intro}
u_t=X\circ u_x
\end{equation}
where $\circ$ is a commutative associative product with unit satisfying Hertling-Manin condition. For such systems the integrability conditions can be written in terms of the Riemann tensor of a linear connection uniquely defined by the system \eqref{F-sys-intro} and in terms of the structure functions of the product:
\beq\label{shc-intro}
R^s_{lmi}c^j_{ks}+R^s_{lik}c^j_{ms}+R^s_{lkm}c^j_{is}=0,\,\text{or}\,R^j_{skl}c^s_{mi}+R^j_{smk}c^s_{li}+R^j_{slm}c^s_{ki}=0.
\eeq
It was proved in \cite{LPR}  that  the above conditions in the  semisimple case reduce  to  Tsarev's semi-Hamiltonian  condition. In the holomorphic setting, assuming regularity, David and Hertling proved in \cite{DH} the existence of a special set of coordinates reducing $X\circ$ to a block-diagonal operator 
\[\text{diag}(V_{(1)} , \dots, V_{(r)}),\qquad r\leq n,\] 
whose generic $\alpha^{\text{th}}$ block, of size $m_\alpha$, is of the lower-triangular Toeplitz form
\begin{equation}\label{toeplitz}
	V_{(\alpha)}=
	\begin{bmatrix}
		v^{1(\alpha)} & 0 & \dots & 0\cr
		v^{2(\alpha)} & v^{1(\alpha)} & \dots & 0\cr
		\vdots & \ddots & \ddots & \vdots\cr
		v^{m_\alpha(\alpha)} & \dots & v^{2(\alpha)} & v^{1(\alpha)}
	\end{bmatrix}.
\end{equation} 
In  order to  extend Tsarev's theory,  one  is forced 
 to  restrict to  the subclass of  systems (called \emph{Darboux-Tsarev systems}) defined by functions,   $v^{i(\alpha)}$, depending  on a \textit{subset} of the variables $(u^1 , \dots, u^n)$. More precisely, 
\begin{equation}\label{DT-intro}
v^{i(\alpha)}=v^{i(\alpha)}(u^{1(1)} , \dots, u^{k_1(1)} , \dots, u^{1(r)} , \dots, u^{k_r(r)}),
\end{equation}
where $k_s=m_s$ if $i>m_s$ and  $k_s=i$ otherwise. It  turns out that  these conditions are sufficient to prove the completeness of the symmetries in the hodograph formula and the existence of the metrics related to  Hamiltonian formalism (see \cite{LPVG2,LPVG3}). 
 These results rely on a classical Darboux theorem applied to the linear system for the symmetries and to Dubrovin-Novikov-Tsarev system for the metrics respectively. In the recent paper \cite{AL} we proved  that  in the analytic/holomorphic setting Darboux-Tsarev restrictions are no longer necessary for the completeness of the symmetries. In the same setting it is possible to prove a generalization of Novikov's  conjecture and of Sevennec's result about integrability of conservative systems. These results apply to cyclic F-manifolds (that include regular ones) and are based on 
 the Cauchy–Kovalevskaya theorem combined with the involutivity of the linear system for the symmetries coming  from the conditions \eqref{shc-intro}.
\item Other interesting examples appeared in relation to some classical finite-dimensional integrable systems \cite{BKM1}, in the study  of confluence of hypergeometric functions \cite{KK}, in the study of parabolic regularization of the gradient catastrophes for the Burgers–Hopf equation \cite{KO} and as hydrodynamic reductions of integrable models (see for instance \cite{FP,VF, XF} and references therein).
\end{itemize}

In the recent paper \cite{BKM3} the authors  found ``canonical coordinates'' for tensor fields of type $(1,1)$ with vanishing Haantjes  tensor. It turns out  that in the $\mathfrak{gl}$-regular case, in a neighborhood of an algebraically generic point (together, the two conditions are the analog of David-Hertling's regularity) there exist local coordinates reducing 
 the tensor field to a block-diagonal operator 
\[\text{diag}(V_{(1)} , \dots, V_{(r)}),\qquad r\leq n.\]
In the holomorphic setting the generic $\alpha^{\text{th}}$ block, of size $m_\alpha$, is of the upper-triangular Toeplitz form
\begin{equation}\label{toeplitz-bkm}
	V_{(\alpha)}=
	\begin{bmatrix}
		v^{m(\alpha)} & v^{m-1(\alpha)} & \dots & v^{1(\alpha)}\cr
		0 & v^{m(\alpha)} & \dots & v^{2(\alpha)}\cr
		\vdots & \ddots & \ddots & \vdots\cr
         0 & \dots & 0 & v^{m(\alpha)}
	\end{bmatrix}.
\end{equation} 
For the sake of being self-contained, we recall some terminology here.  The \emph{algebraic type} of an operator field $A$ at a point $q$ is the Segre characteristic of $A_q$, that is the structure of its Jordan canonical form as recorded by the sizes of the Jordan blocks attached to each eigenvalue, the values of the eigenvalues themselves being irrelevant. A point $p\in M$ is \emph{algebraically generic} if the algebraic type of $A$ is the same at all points of some neighborhood of $p$, and \emph{singular} otherwise.  By continuity the algebraically generic points form an open subset of $M$. For a $\mathfrak{gl}$-regular operator field each eigenvalue carries exactly one Jordan block, so algebraic genericity amounts to the requirement that the number of distinct eigenvalues, their multiplicities, and their partition into real eigenvalues and pairs of complex conjugate ones be locally constant, in other words that eigenvalues do not collide nor split near $p$. This is exactly what forces the number $r$ of blocks and their sizes $m_{\alpha}$ in \eqref{toeplitz-bkm} to be locally constant, and it is why the hypothesis appears in the two Propositions below.

Applying a further elementary change of variables one ends  up exactly with the same form obtained by David and Hertling in the case of regular F-systems. This  means  that the main result of \cite{BKM3} can be rephrased in  the following way. 
\begin{proposition}\label{prop1.1}
In a neighborhood of an algebraically generic point, systems of hydrodynamic type, defined by $\mathfrak{gl}$-regular tensor fields of type $(1,1)$ with vanishing Haantjes  tensor, coincide with regular F-systems.
\end{proposition}

In  this paper we study the relation between the integrability and the vanishing of the Haantjes tensor. As far as  we know in all known examples of integrable systems of hydrodynamic type the Haantjes tensor vanishes. For instance:
\begin{itemize}
\item In the case of tensors of the form \eqref{LM} this  follows from the vanishing of Nijenhuis torsion  of $L$ and from the general properties of the Haantjes tensor (see \cite{Bogoyavlenskij}).
\item In the case of F-systems \eqref{F-sys-intro} the vanishing of the Haantjes tensor follows from Hertling-Manin condition (we refer to  \cite{LPR}  for  details).
\end{itemize}
For this reason, it is a common opinion among people working in this field that integrability should be related in some way to the vanishing of the Haantjes tensor.
For instance, Ferapontov and  Marshall proved that  the  vanishing of the  Haantjes tensor is a necessary condition for integrability of conservative hydrodynamic chains \cite{FM}.
In the context of systems of hydrodynamic type, a precise conjecture, supported by several computations, has been recently proposed in \cite{BKMint}  by Bolsinov, Konyaev and Matveev for systems of hydrodynamic type defined by a $\mathfrak{gl}$-regular tensor field.  Broadly speaking, the conjecture asserts that $\mathfrak{gl}$-regularity together with integrability forces the vanishing of the Haantjes torsion. We first state it precisely. 

We recall the terminology of \cite{BKMint}. For operator fields $L,M$ the $(1,2)$-tensor field $\langle L,M\rangle$ is defined by
\[
        \langle L,M\rangle(\xi,\eta):=M[L\xi,\eta]+L[\xi,M\eta]-[L\xi,M\eta]-LM[\xi,\eta],
\]
and $M$ is a \emph{symmetry} of $L$ if $LM=ML$ and the symmetric part of $\langle L,M\rangle$ vanishes, that is $\langle L,M\rangle(\xi,\xi)=0$ for every $\xi$.  The relation of being a symmetry is symmetric in $L,M$, and operators $K_1,\dots,K_n$ are \emph{mutual symmetries} if $K_i$ is a symmetry of $K_j$ for all $i,j$. By \cite[Introduction]{BKMint} this is equivalent to the formal compatibility of the corresponding systems $u_{t_i}=K_i(u)u_x$, that is to the pairwise commutativity of their flows.  In particular $M$ is a symmetry of $L$ if and only if $u_t=L(u)u_x$ and $u_\tau=M(u)u_x$ have commuting flows.

\begin{conjecture}[Bolsinov-Konyaev-Matveev \cite{BKMint}]\label{BKMc}
Let $M$ be an $n$-dimensional manifold and let $K_1,\dots,K_n$ be mutual symmetries on $M$ that are linearly independent at every point. Assume that there exists a linear combination $\sum_ic_iK_i$, with constant coefficients, which is $\mathfrak{gl}$-regular. Then the Haantjes torsion of all $K_i$ vanishes.
\end{conjecture}

A word on the $\mathfrak{gl}$-regularity locus is in order, since it is used systematically below.  We call an operator field \emph{$\mathfrak{gl}$-regular at $p$} in the pointwise sense of Definition~\ref{def:cyclic-pair}, and simply \emph{$\mathfrak{gl}$-regular} when it is $\mathfrak{gl}$-regular at every point of $M$,  the latter is the reading of the conjecture above.  By Lemma~\ref{lem:glregular} the set of points at which an operator field is $\mathfrak{gl}$-regular is open, so the hypothesis is a local one and so are all the arguments below.  What we actually prove is the  statement that $\mathfrak{gl}$-regularity of $A:=\sum_ic_iK_i$ at a \emph{single} point $p$ already forces the Haantjes torsion of $A$, and thus that of all the $K_i$, to vanish on a whole neighborhood of $p$ (Theorem~\ref{thm:conjecture} and Proposition~\ref{prop:reduction}).  Since the Haantjes torsion is a tensor field, and therefore continuous, it is in fact enough that $A$ be $\mathfrak{gl}$-regular at the points of a dense subset of $M$ for the conclusion of the conjecture to hold on all of $M$.

In this paper we prove Conjecture \ref{BKMc} and combining this result with Proposition \ref{prop1.1}, we get  the following important consequence.

\begin{proposition}\label{F-Prop}
In a neighborhood of an algebraically generic point, integrable systems of hydrodynamic type, defined by $\mathfrak{gl}$-regular tensor fields of type $(1,1)$, coincide with regular F-systems satisfying integrability condition \eqref{shc-intro}.
\end{proposition}

Although the natural connection used below originates from the cyclic F-manifold framework of \cite{AL}, all the properties of cyclic pairs required for the proof of the Bolsinov--Konyaev--Matveev conjecture are established independently in the present paper.

The paper is organized as follows: in Section 2 we introduce cyclic pairs and their properties, and in Section 3 we introduce the associated natural connections and their properties. Using these notions, in Section 4 we prove Conjecture \ref{BKMc}. The paper also contains two appendices: the first is devoted to a generalization of Tsarev's method of constructing solutions using symmetries. We show that the method also works for general systems of hydrodynamic type (even non-integrable ones), assuming the existence of a cyclic vector field. In the case of F-systems, the construction coincides with that in \cite{LPVG2} (see also \cite{AL}). In the second appendix, we give an example of a $2$-component system of hydrodynamic type, defined by a tensor field with vanishing Haantjes torsion, which is not reducible to an F-system in a neighborhood of a point that is not algebraically generic.

\section{Cyclic pairs}\label{SectionCyclicPairs}

Throughout, $M$ is a manifold of dimension $n$ in the category of $C^{\infty}$ manifolds, of $C^{\omega}$ real analytic manifolds, or of complex (holomorphic) manifolds, and $\K=\reals$ or $\compl$ accordingly.  We denote by $\O$ the structure sheaf, by $\T$ the tangent sheaf and by $\T^*$ the cotangent sheaf.  By an \emph{operator field} we mean a section $A$ of $\End(\T)$, that is a $(1,1)$-tensor field. For a torsionless connection $\nabla$ and an operator field $A$ we use
\begin{equation}\label{eq:dnablaA}
(d_{\nabla}A)(Y,Z)=\nabla_Y(AZ)-\nabla_Z(AY)-A([Y,Z])=(\nabla_YA)(Z)-(\nabla_ZA)(Y).
\end{equation}
The \emph{Nijenhuis torsion} and the \emph{Haantjes torsion} of an operator field $A$ are the $(1,2)$-tensor fields
\begin{align}
\label{eq:Nijenhuis}
N_A(Y,Z)&:=A^2[Y,Z]+[AY,AZ]-A[AY,Z]-A[Y,AZ],\\
\label{eq:Haantjes}
H_A(Y,Z)&:=A^2N_A(Y,Z)+N_A(AY,AZ)-A\,N_A(AY,Z)-A\,N_A(Y,AZ),
\end{align}
respectively.  $A$ is a \emph{Haantjes operator} if $H_A=0$.
The composition of operator fields is denoted by juxtaposition and $[A,B]:=AB-BA$.  Furthermore,  if $A$ is an operator field,  we denote the induced endomorphism of $T_pM$ by $A_p$ instead of $A(p)$. 

\begin{definition}\label{def:cyclic-pair}
Let $A$ be an operator field on $M$ and let $v$ be a vector field on $M$. The pair $(A,v)$ is a \emph{cyclic pair} on an open set $U\subseteq M$ if
\[
        v,\,Av,\,\dots,\,A^{n-1}v
\]
is a frame of $\T|_U$. In this case we say that $v$ is \emph{cyclic} for $A$ on $U$. A point $p\in M$ is a \emph{cyclic point} of $A$, or $A$ is \emph{$\mathfrak{gl}$-regular at $p$}, if there exists $w\in T_pM$ such that $w,A_pw,\dots,A_p^{n-1}w$ is a basis of $T_pM$.
\end{definition}

The following lemma shows that if $A$ is $\mathfrak{gl}$-regular at $p$ then there exists a neighborhood $V$ of $p$ and a local vector field $v$ such that $(A,v)$ is a cyclic pair.  The proof hinges on the rational canonical form of a matrix and the classification of modules over a principal ideal domain (PID), so it works over any field, bypassing the need to use eigenvectors. 
\begin{lemma}\label{lem:glregular}
Let $A$ be an operator field on $M$ and let $p\in M$. The following are equivalent:
\begin{enumerate}
\item[(i)] $A$ is $\mathfrak{gl}$-regular at $p$, that is, there is $w\in T_pM$ with $w,A_pw,\dots,A_p^{n-1}w$ a basis of $T_pM$;
\item[(ii)] the minimal polynomial of $A_p$ has degree $n$, equivalently coincides with the characteristic polynomial of $A_p$;
\item[(iii)] over an algebraic closure $\overline\K$ of $\K$, every eigenvalue of $A_p$ has geometric multiplicity one.
\end{enumerate}
If these hold, then for every vector field $v$ on $M$ with $v(p)=w$ as in \emph{(i)} there is an open neighborhood $V$ of $p$ such that $(A,v)$ is a cyclic pair on $V$. 
\end{lemma}

\begin{proof}
Write $F:=T_pM$, $L:=A_p\in\End(F)$, and let $\mu$ and $\chi$ be the minimal and characteristic polynomials of $L$.  Both are monic, $\mu\mid\chi$, and $\deg\chi=n$. The equivalence of \emph{(i)} and \emph{(ii)} is well-known (see for instance Theorem 7.9 in \cite{Roman}).  We report it here in the current language for completeness.

\emph{(i)$\Rightarrow$(ii).} If $w,Lw,\dots,L^{n-1}w$ is a basis of $F$ and $q(L)=0$ for a polynomial $q$ of degree ${}<n$, then $q(L)w=0$ expresses a nontrivial linear dependence among $w,Lw,\dots,L^{n-1}w$ unless $q=0$. Thus no monic polynomial of degree ${}<n$ annihilates $L$, so $\deg\mu=n$ and $\mu=\chi$.

\emph{(ii)$\Rightarrow$(i).} Regard $F$ as a module over the principal ideal domain $\K[t]$, the indeterminate $t$ acting as $L$. By the structure theorem for finitely generated modules over a  PID (see \cite{Aluffi, Roman}), there are monic polynomials $d_1\mid d_2\mid\dots\mid d_r$ in $\K[t]$, the invariant factors of $L$, with
\[
        F\;\cong\;\bigoplus_{i=1}^{r}\K[t]/(d_i)
\]
as $\K[t]$-modules.  Moreover $d_r=\mu$ and $d_1\cdots d_r=\chi$. If $\deg\mu=n=\deg\chi$ then $\deg d_r=\sum_i\deg d_i$, forcing $r=1$, so that $F\cong\K[t]/(\chi)$ is a cyclic $\K[t]$-module. Let $w\in F$ correspond to the class of $1$.  Then $\{t^k\cdot w\}_{k\geq0}=\{L^kw\}_{k\geq0}$ spans $F$, and since $\dim_{\K}\K[t]/(\chi)=\deg\chi=n$ the vectors $w,Lw,\dots,L^{n-1}w$ are a basis of $F$. This argument is valid over any field, in particular over $\K=\reals$ and over $\K=\compl$. It uses the rational canonical form and does not require  factorization of $\chi$ or eigenvectors.

\emph{(ii)$\Leftrightarrow$(iii).} Over $\overline\K$ decompose $\chi=\prod_{s}(t-\lambda_s)^{m_s}$ with the $\lambda_s$ distinct. For each $s$ the geometric multiplicity of $\lambda_s$ is $\dim_{\overline\K}\ker(L-\lambda_s\,\Id)$, which equals the number of Jordan blocks of $L$ with eigenvalue $\lambda_s$, while the multiplicity of $(t-\lambda_s)$ in $\mu$ equals the size of the largest such block.  Therefore every eigenvalue has geometric multiplicity one if and only if, for each $s$, there is a single Jordan block, that is if and only if the exponent of $(t-\lambda_s)$ in $\mu$ equals $m_s$ for every $s$. Since $\mu\mid\chi$ this is equivalent to $\mu=\chi$, that is to (ii). The extension of scalars from $\K$ to $\overline\K$ does not change $\mu$ or $\chi$.

Now,  let $v$ be a vector field on $M$ with $v(p)=w$ as in (i). In a local frame near $p$ the vector fields $v,Av,\dots,A^{n-1}v$ have components depending continuously (indeed as smoothly as $A$ and $v$) on the point, and

\[
        \Delta:=\det\bigl[\,v\;\;Av\;\;\cdots\;\;A^{n-1}v\,\bigr]
\]
is a continuous function on a neighborhood of $p$ with $\Delta(p)\neq0$, because $w,Lw,\dots,L^{n-1}w$ is a basis of $F$.  Thus $\Delta\neq0$ on an open neighborhood $V$ of $p$, and on $V$ the sections $v,Av,\dots,A^{n-1}v$ form a frame of $\T|_V$.  Since a vector field $v$ with $v(p)=w$ always exists, the final assertion follows.
\end{proof}

\begin{remark}\label{rmk:glregular}
\normalfont
The equivalent conditions of Lemma \ref{lem:glregular} are the classical characterizations of a non-derogatory, or $\mathfrak{gl}$-regular, endomorphism.  The terminology is the one of \cite{BKM3}. Over $\K=\compl$ a cyclic vector is obtained at once by summing cyclic vectors of the generalized eigenspaces, whose minimal polynomials are pairwise coprime. Over $\K=\reals$ this recipe is unavailable when $A_p$ has a non-real eigenvalue.  The proof of Lemma \ref{lem:glregular}  avoids eigenvectors and produces a cyclic vector directly from the rational canonical form.
\end{remark}

The next lemma collects the linear algebra of a cyclic endomorphism that will be used repeatedly.
\begin{lemma}\label{lem:cyclic-linear-algebra}
Let $V$ be an $n$-dimensional vector space over $\K$, let $A\in\End(V)$ and let $v\in V$ be cyclic for $A$. Denote by
\[
        \K[t]_{<n}:=\{p\in\K[t]\;:\;\deg p<n\},
        \quad
        \K[A]:=\spanop\{\Id,A,\dots,A^{n-1}\}\subseteq\End(V),
\]
and by $\Cyc(A):=\{B\in\End(V)\;:\;[A,B]=0\}$ the centralizer of $A$. Then:
\begin{enumerate}
\item[(i)] the map $\K[t]_{<n}\rightarrow\K[A]$, $p\mapsto p(A)$, is an isomorphism of vector spaces,  in particular $\Id,A,\dots,A^{n-1}$ are linearly independent and $\dim_{\K}\K[A]=n$;
\item[(ii)] for every $Y\in V$ there is a unique polynomial $p_Y\in\K[t]_{<n}$ with
\[
        Y=p_Y(A)\,v,
\]
and $Y\mapsto p_Y$ is a linear isomorphism $V\rightarrow\K[t]_{<n}$.  Equivalently the evaluation map $\Phi_v:\K[A]\rightarrow V$, $\Phi_v(P):=Pv$, is a linear isomorphism, and $p_Y(A)=\Phi_v^{-1}(Y)$;
\item[(iii)] $\Cyc(A)=\K[A]$.
\end{enumerate}
\end{lemma}

\begin{proof}
(i) The map $p\mapsto p(A)$ is linear and surjective onto $\K[A]$ by the definition of $\K[A]$. If $p(A)=0$ with $p=\sum_{k=0}^{n-1}c_kt^k$, then applying $p(A)$ to $v$ gives $\sum_{k=0}^{n-1}c_kA^kv=0$.  Since $v,Av,\dots,A^{n-1}v$ is a basis of $V$, all $c_k$ vanish, so $p=0$ and the map is injective. As $\dim_{\K}\K[t]_{<n}=n$, it is an isomorphism and $\dim_{\K}\K[A]=n$.

(ii) Consider the composite
\[
        \K[t]_{<n}\xrightarrow{\;p\mapsto p(A)\;}\K[A]\xrightarrow{\;\Phi_v\;}V,
        \qquad
        p\longmapsto p(A)\,v .
\]
The first arrow is the isomorphism of (i). The second, $\Phi_v$, is linear. It is surjective because its image contains $\Phi_v(A^k)=A^kv$ for $k=0,\dots,n-1$, hence a basis of $V$, and it is therefore an isomorphism since $\dim\K[A]=n=\dim V$. The composite $p\mapsto p(A)v$ sends $t^k$ to $A^kv$, so it carries the basis $1,t,\dots,t^{n-1}$ of $\K[t]_{<n}$ to the basis $v,Av,\dots,A^{n-1}v$ of $V$, and is thus a linear isomorphism. For $Y\in V$ let $p_Y\in\K[t]_{<n}$ be its preimage: it is the unique polynomial of degree $<n$ with $p_Y(A)v=Y$, and $Y\mapsto p_Y$ is the inverse isomorphism. Finally $p_Y(A)=\Phi_v^{-1}(Y)$, since $\Phi_v(p_Y(A))=p_Y(A)v=Y$.

(iii) The inclusion $\K[A]\subseteq\Cyc(A)$ is immediate, as $A$ commutes with each $A^k$. Conversely let $B\in\Cyc(A)$ and define $P:=p_{Bv}(A)\in\K[A]$, so that $Pv=Bv$ by (ii). For arbitrary $Y\in V$ write $Y=p_Y(A)v$.  Since $B$ commutes with $A$ it commutes with every element of $\K[A]$,  thus
\[
        BY=B\,p_Y(A)\,v=p_Y(A)\,Bv=p_Y(A)\,Pv=P\,p_Y(A)\,v=PY,
\]
where the second to last equality holds because $P\in \K[A]\subseteq\Cyc(A)$.
Therefore $B=P\in\K[A]$.
\end{proof}

\medskip

The next lemma is a further elementary property of cyclic endomorphisms, used twice below: in the proof of Theorem \ref{thm:E2} and in the proof of Lemma \ref{lem:AK-compatible}.  It is the mechanism through which cyclicity upgrades a partial vanishing statement for a vector-valued $2$-form to its vanishing.

\begin{lemma}\label{lem:skew-rigidity}
Let $F$ be an $n$-dimensional vector space over $\K$, let $A\in\End(F)$ and let $v\in F$ be cyclic for $A$. Let $H\in\Lambda^2F^*\otimes F$ satisfy
\begin{equation}\label{eq:skew-rigidity-hyp}
        H(Y,AZ)+H(Z,AY)=0\qquad\forall\,Y,Z\in F .
\end{equation}
Then $H=0$.
\end{lemma}

\begin{proof}
Let $\lambda\in F^*$ be arbitrary and set $h_{\lambda}(Y,Z):=\lambda\bigl(H(Y,Z)\bigr)$, a skew-symmetric bilinear form on $F$.  Pairing \eqref{eq:skew-rigidity-hyp} with $\lambda$ gives $h_{\lambda}(Y,AZ)+h_{\lambda}(Z,AY)=0$, and since $h_{\lambda}(Z,AY)=-h_{\lambda}(AY,Z)$ this reads
\[
        h_{\lambda}(AY,Z)=h_{\lambda}(Y,AZ)\qquad\forall\,Y,Z\in F,
\]
that is, $A$ is self-adjoint with respect to $h_{\lambda}$.  Iterating this identity, for all $i,j\geq0$,
\[
        h_{\lambda}(A^iv,A^jv)=h_{\lambda}(v,A^{i+j}v)=h_{\lambda}(A^{i+j}v,v)=-h_{\lambda}(v,A^{i+j}v),
\]
where the first two equalities move the powers of $A$ across $h_{\lambda}$ and the last one is the skew-symmetry of $h_{\lambda}$.  Therefore $h_{\lambda}(v,A^{i+j}v)=0$, and therefore $h_{\lambda}(A^iv,A^jv)=0$ for all $i,j\geq0$.  Since $v,Av,\dots,A^{n-1}v$ is a basis of $F$, this forces $h_{\lambda}=0$.  As $\lambda\in F^*$ is arbitrary, $H=0$.
\end{proof}

\section{The natural connection of a cyclic pair}\label{SectionNaturalConnection}

\begin{theorem}\label{thm:E1}
Let $(A,v)$ be a cyclic pair on $U$. Then there exists a unique torsionless connection $\nav$ on $U$ such that
\begin{equation}\label{eq:E1-characterization}
        \nav v=0,
        \qquad
        d_{\nav}A=0 .
\end{equation}
\end{theorem}

\begin{proof}
Fix an auxiliary torsionless connection $\bar\nabla$ on $U$.  Any other torsionless connection differs from it by a symmetric tensor, $\nabla_YZ=\bar\nabla_YZ+S(Y,Z)$ with $S\in S^2\T^*|_U\otimes\T|_U$, so both conditions in \eqref{eq:E1-characterization} become conditions on $S$.  The first one, $\nav v=0$, reads $S(Y,v)=-\bar\nabla_Yv$ and is solvable at once: since $v$ is nowhere vanishing there is locally a one-form $\theta$ with $\theta(v)=1$, and
\[
        S_0(Y,Z):=-\theta(Y)\bar\nabla_Zv-\theta(Z)\bar\nabla_Yv+\theta(Y)\theta(Z)\bar\nabla_vv
\]
is symmetric and satisfies $S_0(Y,v)=-\bar\nabla_Yv$.  Thus the connections with $\nabla v=0$ are exactly those with $S=S_0+S_1$ and $S_1$ ranging in
\[
        \mathcal E:=\{S_1\in S^2\T^*|_U\otimes\T|_U\;:\;S_1(v,\cdot)=0\},
\]
which is the residual freedom left to satisfy the second condition.  The latter is affine in $S$: from $(\nabla_YA)(Z)=(\bar\nabla_YA)(Z)+S(Y,AZ)-A\bigl(S(Y,Z)\bigr)$ and the symmetry of $S$, the terms $A\bigl(S(Y,Z)\bigr)$ cancel upon antisymmetrization in \eqref{eq:dnablaA}, leaving
\[
        d_{\nabla}A=d_{\bar\nabla}A+\mathcal M_AS,
        \qquad
        (\mathcal M_AS)(Y,Z):=S(Y,AZ)-S(Z,AY),
\]
so that $d_{\nabla}A=0$ is the linear equation $\mathcal M_AS_1=-d_{\bar\nabla}A-\mathcal M_AS_0$, whose right-hand side is a section of $\Lambda^2\T^*|_U\otimes\T|_U$ already determined by the data.  The heart of the argument is the following consequence of cyclicity: the bundle morphism
\[
        \mathcal M_A:\mathcal E\longrightarrow\Lambda^2\T^*|_U\otimes\T|_U
\]
is a fiberwise isomorphism.  It is enough to treat the scalar-valued case, the vector-valued one following componentwise in a local frame: if $b$ is symmetric with $b(v,\cdot)=0$ and $b(Y,AZ)-b(Z,AY)=0$, then $b(AY,Z)=b(Y,AZ)$, that is $A$ is self-adjoint with respect to $b$,  hence $b(A^iv,A^jv)=b(v,A^{i+j}v)=0$ for all $i,j\geq0$.  Since $v,Av,\dots,A^{n-1}v$ is a frame,  we have $b=0$.  Injectivity together with the dimension count
\[
        \dim\mathcal E=n\Bigl(\frac{n(n+1)}{2}-n\Bigr)=\frac{n^2(n-1)}{2}=\dim\bigl(\Lambda^2\T^*\otimes\T\bigr)
\]
yields the isomorphism.  Therefore $S_1=\mathcal M_A^{-1}\bigl(-d_{\bar\nabla}A-\mathcal M_AS_0\bigr)$ exists and is unique, and it is of the same class as $A$ and $v$, because in local frames $\mathcal M_A$ is represented by a matrix whose determinant is nowhere zero on the cyclic locus. Setting $\nav:=\bar\nabla+S_0+S_1$ proves existence.  Uniqueness follows from the same isomorphism.
\end{proof}

The preceding construction is the cyclic-pair version of the distinguished connection introduced in \cite{AL} for cyclic F-manifolds. In that setting $v=e$ is the unit and $A=X\circ$. The proof above shows that the construction only uses the cyclicity of the pair $(A,v)$, and neither the F-product nor the special role of the unit is needed.

\begin{definition}\label{def:natural-connection-pair}
The connection $\nav$ of Theorem \ref{thm:E1} is called the \emph{natural connection} of the cyclic pair $(A,v)$.
\end{definition}
\begin{theorem}\label{thm:E2}
Let $(A,v)$ be a cyclic pair on $U$, let $\nabla:=\nav$ and let $B$ be an operator field on $U$. Then the following are equivalent:
\begin{enumerate}
\item[(i)] the two systems
\[
        u_t=A(u)u_x,
        \qquad
        u_{\tau}=B(u)u_x
\]
have commuting flows;
\item[(ii)] $[A,B]=0$ and $d_{\nabla}B=0$.
\end{enumerate}
Moreover, if $[A,B]=0$ then there are unique $a_0,\dots,a_{n-1}\in\O(U)$ with $B=\sum_{k=0}^{n-1}a_kA^k$.
\end{theorem}

\begin{proof}
We use the general criterion for commutativity of systems of hydrodynamic type established in \cite[Lemma 3.3 and Proposition 3.4]{LPVG1}. For arbitrary operator fields $P,Q$ on $U$ and an arbitrary torsionless connection $\nabla$, the flows of $u_{\sigma}=P(u)u_x$ and $u_{\tau}=Q(u)u_x$ commute if and only if $[P,Q]=0$ and
\begin{equation}\label{eq:general-criterion}
        (d_{\nabla}P)(Y,QZ)+(d_{\nabla}P)(Z,QY)=(d_{\nabla}Q)(Y,PZ)+(d_{\nabla}Q)(Z,PY)
        \qquad\forall\,Y,Z .
\end{equation}
None of the entries is required to be the operator field of the cyclic pair, a freedom that will be exploited in Lemma \ref{lem:AK-compatible}.  Here we take $(P,Q)=(A,B)$ and $\nabla=\nav$.  Thus, for a cyclic pair, commutativity of the flows is equivalent to $[A,B]=0$ together with
\begin{equation}\label{eq:E2-reduced}
        (d_{\nabla}B)(Y,AZ)+(d_{\nabla}B)(Z,AY)=0\qquad\forall\,Y,Z .
\end{equation}
That \eqref{eq:E2-reduced} already forces $d_{\nabla}B=0$ is a pointwise consequence of cyclicity.  Indeed, fix $p\in U$ and apply Lemma \ref{lem:skew-rigidity} with $F=T_pM$, the endomorphism $A_p$, the cyclic vector $v(p)$ and $H:=(d_{\nabla}B)_p\in\Lambda^2T^*_pM\otimes T_pM$: relation \eqref{eq:E2-reduced} evaluated at $p$ is exactly hypothesis \eqref{eq:skew-rigidity-hyp}, so $(d_{\nabla}B)_p=0$.  Since $p$ is arbitrary, $d_{\nabla}B=0$.  The converse implication is immediate, since $[A,B]=0$ and $d_{\nabla}A=d_{\nabla}B=0$ satisfy the criterion above trivially.

For the last statement, assume $[A,B]=0$. By Lemma \ref{lem:cyclic-linear-algebra}(iii), $B_q\in\Cyc(A_q)=\K[A_q]$ for every $q\in U$, and by Lemma \ref{lem:cyclic-linear-algebra}(i) the operators 
\[\Id,A_q,\dots,A_q^{n-1}\] form a basis of $\K[A_q]$.  Therefore there are unique scalars $a_0(q),\dots,a_{n-1}(q)$ with
\[
        B_q=\sum_{k=0}^{n-1}a_k(q)\,A_q^k .
\]
It remains to prove that $a_0,\dots,a_{n-1}\in\O(U)$. Applying this identity to the vector $v(q)$ gives
\[
        B_q\,v(q)=\sum_{k=0}^{n-1}a_k(q)\,A_q^k\,v(q),
\]
that is, $a_0,\dots,a_{n-1}$ are the components of the vector field $Bv$ in the frame $v,Av,\dots,A^{n-1}v$ of $\T|_U$.  Since $Bv$ is a section of $\T|_U$ of the same class as $A$ and $v$, these components belong to $\O(U)$.
\end{proof}

In the particular case of cyclic F-systems, the preceding characterization is part of the framework developed in \cite{AL}.

The connection $\nabla^{(A,v)}$ in general depends on the pair $(A,v)$ not just on $A$.  A cyclic endomorphism \emph{always} has a dense set of cyclic vectors over an infinite field (indeed the cyclic vectors of $A_p$ are the complement of the hypersurface $\det[\,w\;\;A_pw\;\;\cdots\;\;A_p^{n-1}w\,]=0$.  As this polynomial is not identically zero, being nonzero at any cyclic vector, and $\K$ is infinite, the cyclic vectors form a dense open subset of $T_pM$, of full measure when $\K=\reals$ or $\compl$).  Therefore given $A$,  we have infinite choices of cyclic pairs $(A,v)$,  each with its own connection.  However,   the next corollary shows that the set of symmetries for the system $u_t=A u_x$ \emph{does not} depend on the choice of the cyclic vector $v$.
\begin{corollary}\label{cor:E2-independence}
Let $A$ be an operator field on $U$ and let $v,v'$ be two cyclic vector fields for $A$ on $U$. Then
\[
        \{B\;:\;[A,B]=0,\;d_{\nabla^{(A,v)}}B=0\}
        =
        \{B\;:\;[A,B]=0,\;d_{\nabla^{(A,v')}}B=0\}.
\]
\end{corollary}

\begin{proof}
Although the connections $\nabla^{(A,v)}$ and $\nabla^{(A,v')}$ differ in general, each of the two conditions is, by Theorem \ref{thm:E2}, equivalent to one that does not involve $v$ or $v'$. Indeed, applying the equivalence (i)$\Leftrightarrow$(ii) of Theorem \ref{thm:E2} to the cyclic pair $(A,v)$,
\[
        \{B\;:\;[A,B]=0,\;d_{\nabla^{(A,v)}}B=0\}\]\[
        =\{B\;:\;u_t=A(u)u_x\text{ and }u_{\tau}=B(u)u_x\text{ have commuting flows}\},
\]
and applying the same equivalence to the cyclic pair $(A,v')$ gives
\[
        \{B\;:\;[A,B]=0,\;d_{\nabla^{(A,v')}}B=0\}\]\[
        =\{B\;:\;u_t=A(u)u_x\text{ and }u_{\tau}=B(u)u_x\text{ have commuting flows}\}.
\]
The right-hand sides are literally the same set, and commutativity of the flows of $u_t=A(u)u_x$ and $u_{\tau}=B(u)u_x$ is a property of the two systems alone, involving neither $v$ nor $v'$.  Thus the two left-hand sides coincide. In particular, for a fixed $B$ with $[A,B]=0$ the equation $d_{\nabla^{(A,v)}}B=0$ holds if and only if $d_{\nabla^{(A,v')}}B=0$, even though $\nabla^{(A,v)}B$ and $\nabla^{(A,v')}B$ are in general different.
\end{proof}

\section{The Haantjes torsion of an integrable operator field}\label{SectionBKM}
The goal of this section is to prove the  Conjecture \ref{BKMc} of Bolsinov, Konyaev and Matveev recalled in the Introduction.

We begin by establishing the compatibility between $A$ and each single $K_j$, which is used both in the reduction below and in the proof of Theorem \ref{thm:conjecture}.  Note that the hypotheses of the conjecture provide the pairwise compatibility of the $K_i$ among themselves, whereas what is needed is the compatibility of each $K_j$ with the $\mathfrak{gl}$-regular combination $A$.  The passage from the former to the latter is where the constancy of the coefficients $c_i$ enters, and the natural connection of Theorem \ref{thm:E1} is what makes it clear.

\begin{lemma}\label{lem:AK-compatible}
Let $K_1,\dots,K_n$ and $A:=\sum_{i=1}^nc_iK_i$, with $c_1,\dots,c_n\in\K$ constants, be as in the conjecture.  Let $(A,v)$ be a cyclic pair on an open set $U$ and let $\nabla:=\nav$ be its natural connection.  Then, for every $j$,
\[
        [A,K_j]=0
        \qquad\text{and}\qquad
        d_{\nabla}K_j=0
        \quad\text{on }U,
\]
equivalently, by Theorem \ref{thm:E2}, the flows of $u_t=A(u)u_x$ and $u_{t_j}=K_j(u)u_x$ commute.
\end{lemma}

\begin{proof}
The first assertion $[A,K_j]=0$ is clear since the operators $K_i$ are mutual symmetries.

As for the second assertion, we first  fix $j$.  Since the operators $K_i$ are mutual symmetries, the flows of $u_{t_i}=K_i(u)u_x$ and $u_{t_j}=K_j(u)u_x$ commute for every $i$, so applying \eqref{eq:general-criterion} to the pair $(K_i,K_j)$ yields
\begin{equation}\label{eq:criterion-ij}
        (d_{\nabla}K_i)(Y,K_jZ)+(d_{\nabla}K_i)(Z,K_jY)
        =(d_{\nabla}K_j)(Y,K_iZ)+(d_{\nabla}K_j)(Z,K_iY)
\end{equation}
for every $i$ and all local sections $Y,Z$ of $\T|_U$.  We now multiply \eqref{eq:criterion-ij} by $c_i$ and sum over $i$, treating the two sides separately.

On the left-hand side we use that the operation $B\mapsto d_{\nabla}B$ is $\K$-linear.  Consequently
\[
        \sum_{i=1}^nc_i\,d_{\nabla}K_i
        =d_{\nabla}\Bigl(\sum_{i=1}^nc_iK_i\Bigr)
        =d_{\nabla}A=0,
\]
since  $\nabla=\nav$.  Thus the $c_i$-weighted sum of the left-hand sides of \eqref{eq:criterion-ij} is
\[
        (d_{\nabla}A)(Y,K_jZ)+(d_{\nabla}A)(Z,K_jY)=0 .
\]

On the right-hand side the roles are reversed: the operator field $K_j$ is now fixed, and it is the arguments that vary with $i$.  Since $d_{\nabla}K_j$ is a $(1,2)$-tensor field, it is $\O(U)$-bilinear in its two arguments, so the weighted sum can be moved inside the second slot and there it reassembles $A$:
\[
        \sum_{i=1}^nc_i\,(d_{\nabla}K_j)(Y,K_iZ)
        =(d_{\nabla}K_j)\Bigl(Y,\sum_{i=1}^nc_iK_iZ\Bigr)
        =(d_{\nabla}K_j)(Y,AZ),
\]
and likewise, interchanging $Y$ and $Z$, $\sum_ic_i\,(d_{\nabla}K_j)(Z,K_iY)=(d_{\nabla}K_j)(Z,AY)$.  The $c_i$-weighted sum of \eqref{eq:criterion-ij} therefore reduces to
\begin{equation}\label{eq:reduced-Kj}
        (d_{\nabla}K_j)(Y,AZ)+(d_{\nabla}K_j)(Z,AY)=0
        \qquad\forall\,Y,Z .
\end{equation}

It remains to observe that \eqref{eq:reduced-Kj} forces $d_{\nabla}K_j$ to vanish, and this is where cyclicity is used.  Fix $p\in U$ and apply Lemma \ref{lem:skew-rigidity} with $F=T_pM$, the endomorphism $A_p$, the cyclic vector $v(p)$, and $H:=(d_{\nabla}K_j)_p\in\Lambda^2T^*_pM\otimes T_pM$: the identity \eqref{eq:reduced-Kj} evaluated at $p$ is precisely hypothesis \eqref{eq:skew-rigidity-hyp}, whence $(d_{\nabla}K_j)_p=0$.  Since $p\in U$ is arbitrary, $d_{\nabla}K_j=0$ on $U$.

Finally, $[A,K_j]=0$ and $d_{\nabla}K_j=0$ are condition \emph{(ii)} of Theorem \ref{thm:E2} for the pair $(A,K_j)$, so by the implication \emph{(ii)}$\Rightarrow$\emph{(i)} of that theorem the flows of $u_t=A(u)u_x$ and $u_{t_j}=K_j(u)u_x$ commute.
\end{proof}

\medskip

The following proposition provides a convenient reduction of the conjecture: 
\begin{proposition}\label{prop:reduction}
Let $K_1, \dots,  K_n$ and $A:=\sum_{i=1}^nc_iK_i$ be as in the conjecture and let $W\subseteq M$ be an open set on which $A$ is $\mathfrak{gl}$-regular at every point and $H_A=0$.  Then the Haantjes torsion of every $K_i$ vanishes on $W$.
\end{proposition}
\begin{proof}
Fix $q\in W$.  Since $A$ is $\mathfrak{gl}$-regular at $q$,  by Lemma \ref{lem:glregular} there are an open neighborhood $U\subseteq W$ of $q$ and a vector field $v$ on $U$ such that $(A,v)$ is a cyclic pair on $U$. 
By Lemma \ref{lem:AK-compatible} the flows of $u_t=A(u)u_x$ and $u_{t_j}=K_j(u)u_x$ commute, so Theorem \ref{thm:E2} applies to the pair $(A,K_j)$ and provides unique $g_0,\dots,g_{n-1}\in\O(U)$ with
\[
        K_j=g_0\Id+g_1A+\dots+g_{n-1}A^{n-1}
        \qquad\text{on }U .
\]
Since $H_A=0$, every operator field which is a polynomial in $A$ with functional coefficients again has vanishing Haantjes torsion, by \cite{Bogoyavlenskij} (see also \cite[Proposition 2.3]{BKM3}), a property which requires no $\mathfrak{gl}$-regularity of $A$.  Therefore,  $H_{K_j}|_U=0$.  In particular $H_{K_j}$ vanishes at $q$. Since $q\in W$ and $j$ are arbitrary, the Haantjes torsion of every $K_j$ vanishes on $W$.
\end{proof}

Therefore,  by Proposition~\ref{prop:reduction} it is enough to prove the following theorem to show that the conjecture holds. 

\begin{theorem}\label{thm:conjecture}
Let $K_1, \dots,  K_n$ and $A:=\sum_{i=1}^nc_iK_i$ be as in the conjecture and let $p\in M$ be a point at which $A$ is $\mathfrak{gl}$-regular.  Then $H_A$ vanishes on an open neighborhood of $p$.  Consequently, if $A$ is $\mathfrak{gl}$-regular at every point of $M$, then $H_A=0$ on $M$ and the same conclusion holds if $A$ is $\mathfrak{gl}$-regular merely at the points of a dense subset of $M$.
\end{theorem}

The rest of this section is devoted to proving Theorem~\ref{thm:conjecture}. We start with the following observation that disposes of the $2$-dimensional case. 

\begin{observation}\label{rmk:dim2-haantjes}
It is well-known that, in dimension two, every operator field is Haantjes. So the conjecture holds vacuously for $n=2$.  Indeed, the vanishing of $H_A$ is a local statement, so we may argue on a neighborhood $U$ of an arbitrary point, chosen small enough to carry a nowhere vanishing $2$-form (area form) $\w$. Let $A$ be an operator field on $U$.  Since $\dim M=2$ the bundle $\Lambda^2\T^*|_U$ has rank one, so every $2$-form on $U$ is a function multiple of $\w$.  Since $N_A$ is skew-symmetric it can be thought of as a vector-valued $2$-form and, again because $\dim M=2$, there is a vector field $N$ on $U$ with $N_A(Y,Z)=\w(Y,Z)N$. For the same reason the two alternating $2$-forms $\w(A\cdot,A\cdot)$ and $\w(A\cdot,\cdot)+\w(\cdot,A\cdot)$ are function multiples of $\w$.  Evaluating at a point $q\in U$ on a basis $e_1,e_2$ of $T_qM$ with $\w(e_1,e_2)=1$ identifies the two functions through $\w(A_qe_1,A_qe_2)=\det(A_q)$ and $\w(A_qe_1,e_2)+\w(e_1,A_qe_2)=\Tr(A_q)$, giving
\[
        \w(AY,AZ)=\det(A)\,\w(Y,Z),
        \qquad
        \w(AY,Z)+\w(Y,AZ)=\Tr(A)\,\w(Y,Z).
\]
With these, formula \eqref{eq:Haantjes} gives
\[
        H_A(Y,Z)=\w(Y,Z)A^2N+\det(A)\w(Y,Z)N-\Tr(A)\w(Y,Z)AN
        =\] \[=\w(Y,Z)\bigl(A^2-\Tr(A)A+\det(A)\Id\bigr)N=0,
\]
where $\det(A)$ and $\Tr(A)$ denote the functions $q\mapsto\det(A_q)$ and $q\mapsto\Tr(A_q)$ on $U$, and the last equality holds by the Cayley--Hamilton theorem.  Since the point $q$ was arbitrary and the conclusion $H_A=0$ does not involve $\w$, it holds on all of $M$.  Moreover, in dimension two $A$ is $\mathfrak{gl}$-regular at $p$ if and only if $A_p$ is not a multiple of the identity.  
\end{observation}
The previous observation eliminates the case $n=2$.  In the following we assume $n\geq3$.  This hypothesis is genuinely needed: the argument requires $A^2_p$ to be one of the elements of the basis $\Id_p,  A_p,  A^2_p, \dots,  A^{n-1}_p$ of $\K[A_p]$.

We will need the following elementary covariant expression for the Nijenhuis torsion.

\begin{lemma}\label{lem:nijenhuis-torsionless}
Let $\nabla$ be a torsionless connection on $U$ and let $A$ be an operator field on $U$. Then
\begin{equation}\label{eq:nijenhuis-torsionless}
        N_A(Y,Z)=(\nabla_{AY}A)(Z)-(\nabla_{AZ}A)(Y)+A\bigl((\nabla_ZA)(Y)\bigr)-A\bigl((\nabla_YA)(Z)\bigr)
\end{equation}
for all local sections $Y,Z$ of $\T|_U$.
\end{lemma}

\begin{proof}
Since $\nabla$ is torsionless,
\begin{align*}
[Y,Z]&=\nabla_YZ-\nabla_ZY,\\
[AY,AZ]&=\nabla_{AY}(AZ)-\nabla_{AZ}(AY)\\
&=(\nabla_{AY}A)(Z)+A\nabla_{AY}Z-(\nabla_{AZ}A)(Y)-A\nabla_{AZ}Y,\\
[AY,Z]&=\nabla_{AY}Z-\nabla_Z(AY)=\nabla_{AY}Z-(\nabla_ZA)(Y)-A\nabla_ZY,\\
[Y,AZ]&=\nabla_Y(AZ)-\nabla_{AZ}Y=(\nabla_YA)(Z)+A\nabla_YZ-\nabla_{AZ}Y.
\end{align*}
Substituting into \eqref{eq:Nijenhuis},
\begin{align*}
N_A(Y,Z)
&=A^2\nabla_YZ-A^2\nabla_ZY\\
&\quad+(\nabla_{AY}A)(Z)+A\nabla_{AY}Z-(\nabla_{AZ}A)(Y)-A\nabla_{AZ}Y\\
&\quad-A\nabla_{AY}Z+A\bigl((\nabla_ZA)(Y)\bigr)+A^2\nabla_ZY\\
&\quad-A\bigl((\nabla_YA)(Z)\bigr)-A^2\nabla_YZ+A\nabla_{AZ}Y,
\end{align*}
in which the two terms in $A^2\nabla_YZ$, the two in $A^2\nabla_ZY$, the two in $A\nabla_{AY}Z$ and the two in $A\nabla_{AZ}Y$ cancel in pairs, leaving \eqref{eq:nijenhuis-torsionless}.
\end{proof}

Its form is analogous to the covariant Hertling--Manin identity used for F-manifolds: compare \cite{LPVG1,AL}.
\medskip
From now on in this section $(A,v)$ is a cyclic pair on $U$, $\nabla:=\nav$, and
\begin{equation}\label{eq:S-def}
        S(Y,Z):=(\nabla_YA)(Z).
\end{equation}

The following lemma highlights some useful properties of $S(Y,Z)$ and of $N_A$ using the preceding results. 
\begin{lemma}\label{lem:dnabla-A2}
Let $(A,v)$ be a cyclic pair on $U$ and let $S$ be as in \eqref{eq:S-def}. Then:
\begin{enumerate}
\item[(i)] $S(Y,Z)=S(Z,Y)$ for all local sections $Y,Z$ of $\T|_U$;
\item[(ii)] $N_A(Y,Z)=S(AY,Z)-S(AZ,Y)$;
\item[(iii)] $d_{\nabla}(A^2)=-N_A$.
\end{enumerate}
\end{lemma}

\begin{proof}
(i) By \eqref{eq:E1-characterization} and \eqref{eq:dnablaA}, \[0=(d_{\nabla}A)(Y,Z)=(\nabla_YA)(Z)-(\nabla_ZA)(Y)=S(Y,Z)-S(Z,Y).\]

(ii) By (i), the last two terms of \eqref{eq:nijenhuis-torsionless} are $A\bigl(S(Z,Y)-S(Y,Z)\bigr)=0$, and the first two are $S(AY,Z)-S(AZ,Y)$.

(iii) Since $\nabla_Y(A^2)=(\nabla_YA)A+A(\nabla_YA)$, one has \[\bigl(\nabla_Y(A^2)\bigr)(Z)=S(Y,AZ)+A\bigl(S(Y,Z)\bigr),\] thus by \eqref{eq:dnablaA} and (i)
\begin{eqnarray*}
        \bigl(d_{\nabla}(A^2)\bigr)(Y,Z)
        &=&S(Y,AZ)+A\bigl(S(Y,Z)\bigr)-S(Z,AY)-A\bigl(S(Z,Y)\bigr)\\
        &=&S(Y,AZ)-S(Z,AY).
\end{eqnarray*}
Using (i) $S(Y,AZ)-S(Z,AY)=S(AZ,Y)-S(AY,Z)=-N_A(Y,Z)$ where the last equality holds by (ii).
\end{proof}

Let us remark that point (iii) of Lemma~\ref{lem:dnabla-A2} will be essential for the proof of the conjecture. 
Using point (iii) of Lemma~\ref{lem:dnabla-A2} and Theorem~\ref{thm:E2},  with $B=A^2$ we get immediately the following characterization of $N_A$:

\begin{corollary}\label{cor:NA}
The Nijenhuis torsion of $A$ is exactly the obstruction for $u_{\tau}=(A(u))^2u_x$ to be a symmetry of $u_t=A(u)u_x.$
\end{corollary}

By Theorem \ref{thm:E2}, an operator field $K_j$ for which the systems $u_t=A(u)u_x$ and $u_{t_j}=K_j(u)u_x$ have commuting flows is of the form $K_j=c_0\Id+c_1A+\dots+c_{n-1}A^{n-1}$ with $c_0,\dots,c_{n-1}\in\O(U)$ (depending on $j$ of course), and satisfies $d_{\nabla}K_j=0$. 

The following lemma gives the expression for $d_{\nabla}K_j$:
\begin{lemma}\label{lem:dnablaK}
If $K_j=c_0\Id+c_1A+\dots+c_{n-1}A^{n-1}$ with $c_0,\dots,c_{n-1}\in\O(U),$ then 
\begin{equation}\label{eq:leibniz-split}
        (d_{\nabla}K_j)(Y,Z)
        =\underbrace{\sum_{k=0}^{n-1}\bigl[dc_k(Y)\,A^kZ-dc_k(Z)\,A^kY\bigr]}_{\text{coefficient part}}
        +\underbrace{\sum_{k=0}^{n-1}c_k\,\bigl(d_{\nabla}(A^k)\bigr)(Y,Z)}_{\text{power part}}\;.
\end{equation}
\end{lemma}
\begin{proof}
For a function $c\in\O(U)$ and an operator field $P$ on $U$, the covariant derivative of the operator field $cP$ obeys the Leibniz rule
\begin{eqnarray*}
        \bigl(\nabla_Y(cP)\bigr)(Z)
        &=&\nabla_Y\bigl(c\,PZ\bigr)-c\,P(\nabla_YZ)\\
        &=&(Yc)\,PZ+c\bigl(\nabla_Y(PZ)-P(\nabla_YZ)\bigr)\\
        &=&dc(Y)\,PZ+c\,\bigl(\nabla_YP\bigr)(Z),
\end{eqnarray*}
where $dc(Y)=Yc$. Applying this to each summand of $K_j=\sum_{k=0}^{n-1}c_kA^k$ and summing over $k$,
\[
        \bigl(\nabla_YK_j\bigr)(Z)
        =\sum_{k=0}^{n-1}dc_k(Y)\,A^kZ+\sum_{k=0}^{n-1}c_k\,\bigl(\nabla_Y(A^k)\bigr)(Z).
\]
By \eqref{eq:dnablaA}, $(d_{\nabla}K_j)(Y,Z)=(\nabla_YK_j)(Z)-(\nabla_ZK_j)(Y)$, so subtracting the same expression with $Y$ and $Z$ interchanged gives equation \eqref{eq:leibniz-split}.
\end{proof}

\medskip 

The power part contains the Nijenhuis torsion of $A$: its terms $k=0,1$ vanish, since $d_{\nabla}\Id=0$ and $d_{\nabla}A=0$, and its term $k=2$ is $-c_2N_A$ by Lemma \ref{lem:dnabla-A2}(iii). The coefficient part is, at each point, exactly a tensor of the form \eqref{eq:Lambda-form} of the next lemma, with covectors $\alpha_k=dc_k$. The lemma is a purely algebraic cancellation,  and it shows that every such tensor is annihilated by the operation that produces the Haantjes torsion from the Nijenhuis torsion in \eqref{eq:Haantjes}.  This is one of the mechanisms behind the proof of Theorem \ref{thm:conjecture}.

\begin{lemma}\label{lem:cancellation}
Let $F$ be an $n$-dimensional vector space over $\K$, let $A\in\End(F)$ and let $\alpha_0,\dots,\alpha_{n-1}\in F^*$. Define $N\in\Lambda^2F^*\otimes F$ by
\begin{equation}\label{eq:Lambda-form}
        N(Y,Z):=\sum_{k=0}^{n-1}\bigl[\alpha_k(Y)\,A^kZ-\alpha_k(Z)\,A^kY\bigr],
        \qquad Y,Z\in F .
\end{equation}
Then, the Haantjes-ization of $N$ vanishes identically, that is 
\begin{equation}\label{eq:cancellation}
        A^2N(Y,Z)+N(AY,AZ)-A\,N(AY,Z)-A\,N(Y,AZ)=0
        \qquad\forall\,Y,Z\in F .
\end{equation}
\end{lemma}

\begin{proof}
The left-hand side of \eqref{eq:cancellation} is additive in $N$, so it suffices to prove the identity when $N$ is a single summand of \eqref{eq:Lambda-form}: fix $k\in\{0,\dots,n-1\}$, write $\alpha:=\alpha_k$, and set $N(Y,Z):=\alpha(Y)A^kZ-\alpha(Z)A^kY$. Using $A\,A^k=A^{k+1}$,
\begin{align*}
A^2N(Y,Z)&=\alpha(Y)\,A^{k+2}Z-\alpha(Z)\,A^{k+2}Y,\\
N(AY,AZ)&=\alpha(AY)\,A^{k+1}Z-\alpha(AZ)\,A^{k+1}Y,\\
A\,N(AY,Z)&=\alpha(AY)\,A^{k+1}Z-\alpha(Z)\,A^{k+2}Y,\\
A\,N(Y,AZ)&=\alpha(Y)\,A^{k+2}Z-\alpha(AZ)\,A^{k+1}Y .
\end{align*}
In the combination $A^2N(Y,Z)+N(AY,AZ)-A\,N(AY,Z)-A\,N(Y,AZ)$ the eight terms on the right-hand sides cancel in four pairs.  Therefore \eqref{eq:cancellation} holds for each summand, and the general case follows by summing over $k$.
\end{proof}

\medskip

The next engine in the proof of Theorem~\ref{thm:conjecture} is the following.  Since in the formula \eqref{eq:leibniz-split}, the coefficient part gives a vanishing Haantjes tensor,  we would like the power part to have only the Nijenhuis torsion of $A$.  Of course,  this is not feasible on an open set in general, but for each fixed point $p$, we can select a linear combination with constant coefficients of the symmetries $K_j$, call it $\tilde K,$ such that  the power part of $d_{\nabla} \tilde K$ evaluated at $p$ is just $(d_{\nabla}A^2)_p=-(N_A)_p.$ This is the content of the next lemma. 

\begin{lemma}\label{lem:newbasisK}
Assume $n\geq3$, let $(A,v)$ be a cyclic pair on $U$ and let $\nabla:=\nav$ be its natural connection.  Then for each point $p\in U$ there exist unique constants $a_1, \dots, a_n$, depending on $p$, such that the symmetry $\tilde K:=\sum_{j=1}^n a_j K_j$ satisfies $\tilde K_p=A^2_p$.  Consequently,  the power part of $d_{\nabla} \tilde K$ evaluated at $p$ in formula \eqref{eq:leibniz-split} is just $(d_{\nabla}A^2)_p=-(N_A)_p.$
\end{lemma}
\begin{proof} Fix $p\in U$.  The $\K$-algebra $\K[A_p]$ is a vector space of dimension $n$ over $\K$ with basis $\Id_p,  A_p, \dots, A_p^{n-1}$ by Lemma~\ref{lem:cyclic-linear-algebra}(i).  Since $(K_j)_p$ belongs to the centralizer of $A_p$ for all $j$, it belongs to $\K[A_p]$ by Lemma~\ref{lem:cyclic-linear-algebra}(iii). By assumption, the $K_j$'s are linearly independent at each point,  and since there are $n$ of them,  they form at $p$ another basis of $\K[A_p]$.  Thus there exists a unique linear combination with constant coefficients $a_1, \dots,  a_n$ such that $\tilde K:=\sum_{j=1}^n a_j K_j$ satisfies $\tilde K_p=A^2_p$.

Write now $\tilde K$ in terms of powers of $A$ on $U$, that is $\tilde K=\sum_{i=0}^{n-1}c_iA^i$ with coefficient functions $c_i\in \O(U)$, as is possible by the last assertion of Theorem~\ref{thm:E2}.  Since $n\geq3$, $A^2_p$ is itself one of the basis elements $\Id_p,A_p,\dots,A^{n-1}_p$ of $\K[A_p]$, so that the uniqueness of the expansion of $\tilde K_p=A^2_p$ in that basis forces
\[
        c_i(p)=\delta_{2i},\qquad i=0,\dots,n-1 .
\]
  Substituting $c_i(p)=\delta_{2i}$ in the power part of \eqref{eq:leibniz-split} and using $d_{\nabla}\Id=0$, $d_{\nabla}A=0$ and Lemma~\ref{lem:dnabla-A2}(iii) leaves exactly $(d_{\nabla}A^2)_p=-(N_A)_p$, which is the last assertion.
\end{proof}

\medskip 

We are now ready to prove Theorem~\ref{thm:conjecture} and thus the Conjecture \ref{BKMc}: 

\medskip

\begin{proof} Without loss of generality,  we can assume $n\geq3$. Under this assumption $A^2_q$ is always an element of the basis $\Id_q,  A_q, A^2_q, \dots, A_q^{n-1} $ of $\K[A_q]$, at every point $q$ at which $A$ is $\mathfrak{gl}$-regular.

Since $A$ is $\mathfrak{gl}$-regular at $p$,  by Lemma~\ref{lem:glregular} there exist an open neighborhood $U$ of $p$ and a vector field $v$ on $U$ such that $(A,v)$ is a cyclic pair on $U$, with associated torsionless connection $\nabla:=\nabla^{(A, v)}$.  Thus $A$ is $\mathfrak{gl}$-regular at every point of $U$.  By Lemma~\ref{lem:AK-compatible}, $[A,K_j]=0$ and $d_{\nabla}K_j=0$ on $U$, for every $j$.  

Fix now an arbitrary point $q\in U$ and take the constant coefficient linear combination $\tilde K=\sum_{j=1}^n a_jK_j$ provided by Lemma~\ref{lem:newbasisK} at $q$.  Write $\tilde K=\sum_{i=0}^{n-1}c_iA^i$ for coefficient functions $c_i\in \O(U)$,  so that $c_i(q)=\delta_{2i}$. Again, since this is a constant coefficient linear combination of the $K_j$'s, we have $d_{\nabla} \tilde K=0$ on $U$.   Evaluating $d_{\nabla}\tilde K=0$ at $q$,  using \eqref{eq:leibniz-split} together with Lemma~\ref{lem:newbasisK},  we get 
\begin{equation}\label{eq:important}(N_A)|_q(Y,Z)=\sum_{k=0}^{n-1}\bigl[dc_k|_q(Y)\,A_q^kZ-dc_k|_q(Z)\,A_q^kY\bigr],  \quad \forall Y,Z\in T_qM.\end{equation}
Therefore $N_A|_q$ is of the form \eqref{eq:Lambda-form}, with $F=T_qM$, the operator $A_q$ and the covectors $\alpha_k=dc_k|_q$. Since $N_A$ is a $(1,2)$-tensor field, each of the four terms on the right-hand side of \eqref{eq:Haantjes} is $\O$-bilinear in $(Y,Z)$,  therefore the value of $H_A$ at $q$ is, for all $ Y,Z\in T_qM$, $H_A|_q(Y,Z)=$
\[
        A_q^2\,N_A|_q(Y,Z)+N_A|_q(A_qY,A_qZ)-A_q\,N_A|_q(A_qY,Z)-A_q\,N_A|_q(Y,A_qZ),
\]
which vanishes by Lemma \ref{lem:cancellation}. Since $q\in U$ is arbitrary, $H_A=0$ on the neighborhood $U$ of $p$.

The two remaining assertions follow at once.  If $A$ is $\mathfrak{gl}$-regular at every point of $M$, then $M$ is covered by such neighborhoods and $H_A=0$ on $M$.  If $A$ is $\mathfrak{gl}$-regular at the points of a dense subset $D\subseteq M$, then $H_A$ vanishes on the union of the corresponding neighborhoods, which is an open set containing $D$ and hence dense in $M$.  Since $H_A$ is a tensor field, and therefore continuous, $H_A=0$ on $M$.
\end{proof}

\medskip

Combining Theorem~\ref{thm:conjecture} with Proposition~\ref{prop:reduction} proves the conjecture of Bolsinov, Konyaev and Matveev.  Moreover, Theorem~\ref{thm:conjecture} and its proof provide immediately the following corollary:
\begin{corollary}\label{cor:system}
Let $u_t=A(u)u_x$ be a system of hydrodynamic type equipped with a symmetry $u_{\tau}=B(u)u_x$, and assume $n\geq3$ (for $n=2$ the conclusion is automatic by Observation~\ref{rmk:dim2-haantjes}).  If at a point $p$ the operator $A_p$ is $\mathfrak{gl}$-regular and at the same point $B_p=A^2_p$,  then $(H_A)_p=0.$
\end{corollary}

Some concluding comments are in order.  In \cite{BKMint} the conjecture is established in the case in which the $\mathfrak{gl}$-regular combination $\sum_ic_iK_i$ has $n$ distinct real (or complex) eigenvalues, in every dimension, with the stronger conclusion that the $K_i$ are simultaneously diagonalizable, that is that Riemann invariants exist (\cite[Theorem 1]{BKMint}). The vanishing of the Haantjes torsion alone is moreover verified by computer algebra when $K_i$ is conjugate at every point to a single nilpotent Jordan block, in dimensions $n=3,\dots,10$, and when the Jordan form of $K_i$ has two blocks with distinct eigenvalues, in dimensions $n=3,\dots,7$. Theorem \ref{thm:conjecture} completely disposes of these restrictions.  
Theorem \ref{thm:conjecture} does not assume that $p$ is algebraically generic: the proof uses only Theorem \ref{thm:E1} and Theorem \ref{thm:E2}, which hold for an arbitrary cyclic pair, and by Lemma \ref{lem:glregular} cyclic pairs exist near every point at which $A$ is $\mathfrak{gl}$-regular.

The notion of integrability used here is  the existence of $n$ mutually commuting symmetries which are linearly independent at every point.  The mutual commutativity need not be assumed, since it is automatic here.  

\section{Appendix 1: cyclic pairs and the generalized  hodograph method}
In this paper  we  have shown that integrable systems of hydrodynamic type, under mild assumptions, coincide with F-systems satisfying  the integrability  condition \eqref{shc-intro}.  For F-systems (not necessarily integrable) a suitable extension of Tsarev's integration procedure (generalized hodograph method)  has been introduced in  \cite{LPVG2}. In the same paper, the method was applied to prove the solvability of  the Cauchy problem for Darboux-Tsarev systems and in \cite{AL} to prove the solvability of  the Cauchy problem for analytic integrable cyclic F-systems. The idea of the method, starting from the seminal paper of Tsarev \cite{Tsarev}, is to exploit symmetries (commuting flows) to construct solutions. In this  Appendix we show that the same idea works also for general (even not integrable) systems of hydrodynamic type $u_t=A(u)u_x$ assuming the existence of a vector field $v$ defining a cyclic pair $(A,v)$.  Suppose that $u_\tau=B(u)u_x$ is a commuting  flow ($u_{t\tau}=u_{\tau t}$), define  $X=Av$, $Y=Bv$ and assume that the $(1,1)$-tensor field $\tilde{M}$ of components 
\[\tilde{M}=t\nabla X-\nabla Y=t\nabla_v A-\nabla_v B\]
is locally invertible. Consider the algebraic system
\begin{equation}\label{hodographM}
 xI+tA(u)-B(u)=0,
 \end{equation}
 and the  subsystem obtained  applying both sides of \eqref{hodographM} to $v$:
\begin{equation}\label{hodograph}
 xv(u)+tX(u)-Y(u)=0.
 \end{equation}
Clearly the solutions of the system \eqref{hodographM} are also solutions  of the system \eqref{hodograph}. It is immediate to check that the converse statement
 is true due to the commutativity of $A$ and $B$ and the cyclicity of the pair $(A,v)$.
  Indeed, taking into account that $A$ commutes with the left-hand side of \eqref{hodographM}, we have
  \[(xI+tA-B)A^kv=A^k(xI+tA-B)v.\]
  This means that if \eqref{hodograph} holds true, then also \eqref{hodographM} holds true since the left hand side of \eqref{hodographM} applied to the vector fields of the frame generated by the cyclic pair $(A,v)$ vanishes. In other words, the existence of a cyclic pair allows to replace the matrix system of $n^2$ equations \eqref{hodographM} with the  vector system of $n$ equations \eqref{hodograph}. 
  
Taking the derivative of both sides of \eqref{hodograph} with  respect to $x$ and $t$ we get
\begin{equation}\label{hod1}
v+Mu_x=0,\qquad X+Mu_t=0,
\end{equation}
where $M^i_j=x\partial_jv^i+t\partial_jX^i-\partial_jY^i$. It is easy to prove that on the solutions  of the system \eqref{hodographM} $\tilde{M}=M$. Indeed
\[\tilde{M}-M=-x\nabla v=0.\]
Thus on the solutions  of the system \eqref{hodographM} we have
\begin{equation}\label{hod2}
v+\tilde{M}u_x=0,\qquad X+\tilde{M}u_t=0.
\end{equation}
Multiplying both sides of the first equation of \eqref{hod2} by $A$ we obtain
\[X+A\tilde{M}u_x=0,\qquad X+\tilde{M}u_t=0.\]
Let us prove that $A$ and $\tilde{M}$ commute. Using \eqref{hodographM} we obtain
\begin{eqnarray*}
A^i_s\tilde{M}^s_j&=&A^i_s(t\nabla_v A^s_j-\nabla_v B^s_j)=(B^i_s-x\delta^i_s)\nabla_v A^s_j-A^i_s\nabla_v B^s_j\\
\tilde{M}^i_sA^s_j&=&(t\nabla_v A^i_s-\nabla_v B^i_s)A^s_j=\nabla_v A^i_s(B^s_j-x\delta^s_j)-\nabla_v B^i_sA^s_j.
\end{eqnarray*}
This implies
\[[\tilde{M},A]=\nabla_v[A,B]=0.\]
We can conclude that, on the solutions of the system \eqref{hodographM}
\[X+\tilde{M}Au_x=0,\qquad X+\tilde{M}u_t=0\]
and thus
\[\tilde{M}(u_t-Au_x)=0,\]
that  implies $u_t=A(u)u_x$ due to invertibility of $\tilde{M}$. In the case $A=\circ$ where $\circ$  satisfies Hertling-Manin condition the  above procedure coincides with the procedure discussed in \cite{LPVG2,AL}  with $v$ given by the unit vector field.  According to the results of the present paper, in the integrable case and in a neighborhood of an algebraically generic point, if the tensor field of type $(1,1)$ $A$ is $\mathfrak{gl}$-regular this is always the case.

\section{Appendix 2: Representability via an F-product}
Fix an open connected set $U$ of $M$.  Given a cyclic pair $(A,v)$ on $U$,  it is immediate to see that it induces a commutative associative product $\circ_v$ on $\mathfrak{X}(U)$.  Indeed, since any $X\in \mathfrak{X}(U)$ can be written uniquely as $X=\sum_{k=0}^{n-1}f_k A^k v$, $f_k \in \mathcal{O}_M(U)$,  we can define the multiplication operator by $X$ as 
\[ p_X(A):=X\circ_v=\sum_{k=0}^{n-1}f_k A^k,\]
where $p_X(A)$ is precisely the polynomial in $A$ associated with $X$ by Lemma \ref{lem:cyclic-linear-algebra}.
It is clear that $X\circ_v v=X$ and that \[X\circ_v Y=\sum_{k=0}^{n-1}f_k A^k\sum_{j=0}^{n-1}g_j A^j v=\sum_{j=0}^{n-1}g_j A^j\sum_{k=0}^{n-1}f_k A^k v=Y\circ_v X.\]
Associativity follows from the associativity of composition of powers of $A$.

Recall that a commutative associative product on vector fields with unit $e$ satisfies the Hertling–Manin condition iff (definition)
\begin{equation}\label{eq:HM}
        \mathcal{L}_{X\circ Y}(\circ)=X\circ\mathcal{L}_Y(\circ)+Y\circ\mathcal{L}_X(\circ)
        \qquad\text{for all }X,Y ,
\end{equation}
where
\begin{equation}\label{eq:LieDerprod}\mathcal{L}_U(\circ)(Z,W):=[U,Z\circ W]-[U,Z]\circ W-Z\circ[U,W].\end{equation}

For a cyclic pair $(A,v)$ with canonical product $\circ_v$ and unit $v$, we call $v$ a {\bf good unit} if $\circ_v$ satisfies the Hertling–Manin condition. If this is the case,  $\circ_v$ is called an {\bf F-product}.

It is natural to ask whether, given a $\mathfrak{gl}$-regular operator $A$ on an open connected set $U$, there always exists a cyclic vector $v$ such that $\circ_v$ is an F-product.  We show with a simple example discussed in detail that this is not always the case,  even if further conditions are imposed on $A$.  Indeed,  we are going to present a $\mathfrak{gl}$-regular operator $A$ over all $\K^2$ (with coordinates $(u^1, u^2)$),  with vanishing Haantjes tensor,  such that the system $u_t=A(u)u_x$ has two symmetries  $B_1,  B_2$ that are pointwise linearly independent everywhere,  and yet if $U$ is any connected open neighborhood that intersects $\Sigma:=\{u^2=0\},$ then over $U$ $A$ does not admit any cyclic vector that is also a good unit. 

Consider the system
\begin{equation}\label{eq:system}
        u^1_t=u^2_x,
        \qquad
        u^2_t=u^2u^1_x,
        \qquad\text{that is}\qquad
        A=\begin{pmatrix}0&1\\u^2&0\end{pmatrix} ,
\end{equation}
so that
\begin{equation}\label{eq:frame}
        A\partial_1=u^2\,\partial_2,
        \qquad
        A\partial_2=\partial_1,
        \qquad
        A^2=u^2\,\Id.
\end{equation}
The eigenvalues are $\pm\sqrt{u^2}$.  They are real and distinct when $u^2>0$ and complex when $u^2<0$.  The following highlights the properties of $A$:

\begin{proposition}
The operator field \eqref{eq:system} is $\mathfrak{gl}$-regular everywhere, has vanishing Haantjes torsion, and $\Id$ together with $A$ are pointwise independent symmetries of $u_t=Au_x$, with $A$ expressible as a constant-coefficient linear combination of them.
\end{proposition}

\begin{proof}
The operator $A$ is $\mathfrak{gl}$-regular: for $u^2\neq 0$ it is semisimple with two distinct eigenvalues, while at $u^2=0$ it consists of a single Jordan block for the eigenvalue $0$. In two dimensions every operator field is automatically a Haantjes operator. Moreover, $\Id$ and $A$ are symmetries of $u_t=Au_x$ and they are pointwise independent because $A_p$ is never a scalar multiple of the identity.  Finally, $A=0\cdot\Id+1\cdot A$ is a constant-coefficient combination of these symmetries.
\end{proof}

The following theorem shows that being cyclic and being a good unit are incompatible conditions on a connected open set $U$ with $U\cap \Sigma\neq \emptyset.$

\begin{theorem}\label{thm:failure product rep}
Let $U$ be an open set that intersects $\Sigma$. Then there is no cyclic vector field $v$ for $A$ on $U$ that is also a good unit for $\circ_v$.
\end{theorem}

\begin{proof}
Let $v$ be cyclic on $U$, and define for simplicity $w:=Av$ and $\circ:=\circ_v$. In the frame $(v,w)$, the product is determined by
\begin{equation}\label{eq:structure}
        v\circ v=v,
        \qquad
        v\circ w=w,
        \qquad
        w\circ w=A^2v=u^2\,v ,
\end{equation}
where the last equality uses \eqref{eq:frame}: indeed, $C_w$ commutes with $A$ and sends $v$ to $Av$,  thus $C_w=A$ by cyclicity, so $w\circ w=Aw=A^2v$.

Taking $X=Y=v$ in \eqref{eq:HM} and recalling that $v$ is the identity of the canonical product gives $\mathcal{L}_v(\circ)=2\,\mathcal{L}_v(\circ)$,  therefore
\begin{equation}\label{eq:unit-invariant}
        \mathcal{L}_v(\circ)=0 .
\end{equation}

 Write $m:=[v,w]$. Using \eqref{eq:structure}, we have $[v,fv]=v(f)\,v$, and from \eqref{eq:LieDerprod}
\[
        \mathcal{L}_v(\circ)(w,w)
        =[v,w\circ w]-2\,w\circ[v,w]
        =[v,u^2\,v]-2\,w\circ m
        =v(u^2)\,v-2\,w\circ m ,
\]
so \eqref{eq:unit-invariant} yields
\begin{equation}\label{eq:key}
        2\,w\circ m=v(u^2)\,v .
\end{equation}

Now we expand $m=\alpha v+\beta w$ in the chosen frame. Then
\[
        w\circ m=\alpha\,w+\beta\,(w\circ w)=\alpha\,w+\beta u^2\,v ,
\]
which on $\Sigma$, where $u^2=0$, reduces to $\alpha w$. Thus on $\Sigma$ the left-hand side of \eqref{eq:key} lies in $\langle w\rangle$, while the right-hand side lies in $\langle v\rangle$. Since $v$ and $w$ are independent, both sides must vanish, and in particular
\begin{equation}\label{eq:tangency}
        v(u^2)=0
        \qquad\text{on }\Sigma .
\end{equation}

Now we show that this contradicts cyclicity.  Write $v=a\,\partial_1+b\,\partial_2$. Then $v(u^2)=b$, so \eqref{eq:tangency} forces $b=0$ on $\Sigma$. On the other hand, $w=Av=b\,\partial_1+au^2\,\partial_2$ by \eqref{eq:frame}, and therefore
\[
        \det(v,w)=a^2u^2-b^2,
        \qquad
        \det(v,w)\big|_{\Sigma}=-b^2 .
\]
Cyclicity requires $\det(v,w)\neq0$, i.e. $b\neq0$ on $\Sigma$, which contradicts $b=0$. This proves the theorem.
\end{proof}

The mechanism that produces the contradiction is the coincidence of two lines. On $\Sigma$ the operator is nilpotent, $A|_\Sigma=N=\begin{pmatrix}0&1\\0&0\end{pmatrix}$, with $\ker N=\operatorname{Im}N=\langle\partial_1\rangle$, and this line is tangent to $\Sigma=\{u^2=0\}$. Condition \eqref{eq:tangency} says that the unit must be tangent to $\Sigma$ to be a good unit,   thus it lies in $\langle\partial_1\rangle$. On the other hand, cyclicity demands that it avoids $\ker N=\langle\partial_1\rangle$, since any vector in the kernel is an eigenvector of $A|_\Sigma$ and hence cannot be cyclic. These two requirements are incompatible precisely because the nilpotent line is tangent to the singular locus.

For the more general family
\[
        A_c=\begin{pmatrix}0&1\\c&0\end{pmatrix},
        \qquad
        c\in\mathcal{O}(M),
        \qquad
        dc\neq0\ \text{ on }\ \Sigma:=\{c=0\} ,
\]
we have $A_c^2=c\,\Id$, so $A_c$ is $\mathfrak{gl}$-regular everywhere.  Moreover, $A_c|_\Sigma=N$ and $\ker N=\langle\partial_1\rangle$.  The steps above in the proof of Theorem \ref{thm:failure product rep}   carry over and give $v(c)=0$ on $\Sigma$, while the last step in the proof gives $\det(v,A_cv)|_\Sigma=-b^2$, so cyclicity requires $b\neq0$. Writing $v=a\,\partial_1+b\,\partial_2$, the condition $v(c)=0$ reads $a\,\partial_1c+b\,\partial_2c=0$ on $\Sigma$. Therefore no good unit exists precisely when $\partial_1c=0$ on $\Sigma$, i.e. when  $\ker N$ is tangent to $\Sigma$: in that case $\partial_2c\neq0$ and the condition forces $b=0$.

\end{document}